\documentclass[superscriptaddress,twocolumn,aps,pra,10pt]{revtex4-2}

\usepackage{amsmath}
\usepackage{amssymb}
\usepackage{amsthm}
\usepackage{hyperref}
\usepackage{braket}
\usepackage{mathtools}
\usepackage{bm}
\usepackage{bbm}
\usepackage{multirow}

\usepackage{varwidth}
\usepackage{enumitem}
\usepackage{physics}

\usepackage{tikz}
\usetikzlibrary{quantikz2}
\usetikzlibrary{matrix}

\DeclareMathAlphabet{\mathlcal}{U}{dutchcal}{m}{n}

\newtheorem{theorem}{Theorem}
\newtheorem{corollary}{Corollary}
\theoremstyle{remark}
\newtheorem{remark}{Remark}

\newcommand{\ham}{\mathcal{H}}
\newcommand{\rme}{\mathrm{e}}

\begin{document}

\preprint{APS/123-QED}

\title{Quadratically Shallow Quantum Circuits for Hamiltonian Functions}

\author{Youngjun Park}
\affiliation{Department of Chemistry, Yonsei University, Seoul 03722, Republic of Korea}

\author{Minhyeok Kang}
\affiliation{SKKU Advanced Institute of Nanotechnology (SAINT), Sungkyunkwan University, Suwon 16419, Republic of Korea}

\author{Chae-Yeun Park}
\affiliation{School of Integrated Technology, Yonsei University, Seoul 03722, Republic of Korea}
\affiliation{Department of Quantum Information, Yonsei University, Incheon 21983, Republic of Korea}

\author{Joonsuk Huh}
\email{joonsukhuh@yonsei.ac.kr}
\affiliation{Department of Chemistry, Yonsei University, Seoul 03722, Republic of Korea}
\affiliation{Department of Quantum Information, Yonsei University, Incheon 21983, Republic of Korea}

\begin{abstract}
    Many quantum algorithms for ground-state preparation and energy estimation require the implementation of high-degree polynomials of a Hamiltonian to achieve better convergence rates.
    Their circuit implementation typically relies on quantum signal processing (QSP), whose circuit depth is proportional to the degree of the polynomial.
    Previous studies exploit the Chebyshev polynomial approximation, which requires a Chebyshev series of degree \(O(\sqrt{n \ln(1/\delta)})\) for an \(n\)-degree polynomial, where \(\delta\) is the approximation error.
    However, the approximation is limited to only a few functions, including monomials, truncated exponential, Gaussian, and error functions.
    In this work, we present the most generalized function approximation methods for \(\delta\)-approximating linear combinations or products of polynomial-approximable functions with quadratically reduced-degree polynomials.
    We extend the list of polynomial-approximable functions by showing that the functions of cosine and sine can also be \(\delta\)-approximated by quadratically reduced-degree Laurent polynomials.
    We demonstrate that various Hamiltonian functions for quantum ground-state preparation and energy estimation can be implemented with quadratically shallow circuits.
\end{abstract}

\maketitle

\section{Introduction}
Efficient estimation of the ground state and its energy for understanding the physical or chemical properties of quantum systems has been extensively studied in quantum computation \cite{abrams1999quantum,lidar1997simulating,lloyd1996universal,aspuru2005simulated,abrams1997simulation,motta2020determining,aulicino2022state,he2022quantum,ge2019faster,xie2022power,kirby2023exact,seki2021quantum,bespalova2021hamiltonian,epperly2022theory,fomichev2024initial,lin2022heisenberg}.
Many proposed quantum algorithms for ground-state preparation and energy estimation demand synthesizing a polynomial of a Hamiltonian \(\ham\), \(P(\ham)\).
For example, quantum Krylov subspace methods aim to construct a minimax polynomial filter near the ground state energy \cite{epperly2022theory},
and \(\cos^n \ham\) filtering projects an initial state to the ground state \cite{ge2019faster}.
As their performance depends on the overlap between an initial state and the ground state, preparing a good initial state with high-degree polynomial filtering is critical \cite{lin2022heisenberg}.

A higher-degree polynomial of \(\ham\) provides a better convergence to the ground state and its energy.
Its quantum circuit implementation using quantum signal processing (QSP) \cite{low2016methodology,low2017optimal}
requires a deep circuit whose depth is proportional to the degree of \(P(\ham)\), \(\deg(P)\).
In previous works \cite{motta2020determining,aulicino2022state,he2022quantum,ge2019faster,xie2022power,motta2024subspace,kirby2023exact,zhang2024measurement,seki2021quantum,bespalova2021hamiltonian,kyriienko2020quantum}, target Hamiltonian functions are approximated using the Fourier transform or the Fourier approximation using a finite-difference method. However, the Fourier transform involves discretization and truncation errors, and the Fourier approximation using a finite-difference method often exhibits poor convergence rates.

Another approach \cite{gilyen2019quantum,low2017hamiltonian,tosta2024randomized,chakraborty2024implementing,martyn2021grand} uses Chebyshev polynomial approximation \cite{sachdeva2014faster} to reduce \(\deg(P)\).
The polynomial approximation is an optimal \(\delta\)-approximation to monomials on \([-1,1]\) and exponential functions \(\rme^{-x}\) on \([0,b]\) \cite[Chapter 5]{sachdeva2014faster},
which requires a quadratically reduced-degree Chebyshev series.
It can also be used for approximating Gaussian and error functions \cite{low2017hamiltonian}.
However, the aforementioned functions are the only cases known to admit quadratically reduced-degree approximations with a uniform error \(\delta\) on \([-1,1]\).

Generalized quantum signal processing (GQSP) \cite{de2022fourier,motlagh2024generalized,haah2019product} introduces an alternative scheme to QSP by adopting SU(2) ancillary rotations and Hamiltonian simulations controlled on the ancilla qubit.
A polynomial \(P(\rme^{i\ham})\) or a Laurent polynomial \(L(\rme^{i\ham})\) can be implemented by GQSP, requiring \(O(\deg (P))\) or \(O(\deg(L))\) controlled Hamiltonian simulations.
As the Suzuki-Trotter decomposition enables ancilla-free Hamiltonian simulations, GQSP stands out as a promising framework for implementing Hamiltonian functions.
However, efficient function approximation theory to \(P(\rme^{i\ham})\) and \(L(\rme^{i\ham})\) is notably underdeveloped.

In this work, we develop the most generalized function approximation methods that require only quadratically reduced-degree polynomials for linear combinations or products of polynomial-approximable functions.
The approximations are general in the sense that target functions are defined in an extended domain \([-1,1] \times \mathbb{R}\).
We refer to the functions as \textit{polynomial-approximable functions} when either the Chebyshev polynomial approximation or the \textit{Laurent polynomial approximation} is possible.
We newly formulate the Laurent polynomial approximation, which extends the Chebyshev polynomial approximation to Laurent polynomial representations.
This work enables the synthesis of various Hamiltonian functions using quadratically shallow circuits, without altering the total error scaling.

The remainder of the paper is structured as follows.
In Section~\ref{sec:poly_approx}, we review the Chebyshev polynomial approximation \cite{sachdeva2014faster,low2017hamiltonian}, and then present our novel Laurent polynomial approximation.
Table~\ref{tab:poly_approx} summarizes Section~\ref{sec:poly_approx}, and lists the polynomial-approximable functions.
In Section~\ref{sec:lc_prod}, we establish the main contributions of this work: 
Theorems~\ref{thm:linear_comb} and~\ref{thm:product} show that linear combinations and products of the polynomial-approximable functions can be approximated by quadratically reduced-degree polynomials, respectively.
The theorems are proven in Appendix~\ref{appendix:lc_prod_proof}.
In Section~\ref{sec:synthesizing}, we apply the theorems to various Hamiltonian functions.
We achieve a quadratic reduction in the depths of the QSP circuit in the dependence on \(\deg(P)\): from \(O(\deg(P))\) to \(O\bigl(\sqrt{\deg(P) \ln(1/\delta)}\bigr)\).
Section~\ref{sec:synthesizing} is summarized in Table~\ref{tab:hamil_function}, which shows that a broad class of Hamiltonian functions can be efficiently synthesized, including those potentially useful in quantum ground-state preparation and energy estimation.
Section~\ref{sec:conclusion} concludes the paper.
Appendices~\ref{appendix:proof}--\ref{appendix:analysis_gqsp_circ} provide supplementary reviews and details.

\section{The Chebyshev and Laurent Polynomial Approximation}
\label{sec:poly_approx}
\begin{table*}[t]
    \begin{tabular}{|c|c||c|c|c||c|}
        \hline
        Function & (Trunc.) Degree & Degree of Approx. & Approx. Formula & Error & Ref. 
        \\ \hline
        \(x^n\), \(x \in [-1,1]\) & \(n\) & \(O\bigl(\sqrt{n \ln(1/\delta)}\bigr)\) & Eq.~\eqref{eq:mono_approx} 
        & \(\delta > 0\) & \cite{sachdeva2014faster}
        \\ \hline
        \(\exp(-\beta (1+x)), \beta > 0\) & \(t_{\beta}\) & \(O\bigl(\sqrt{t_{\beta}\ln(1/\delta)}\bigr)\) 
        & Eq.~\eqref{eq:exp_approx} & \(\delta \in (0, 1/2]\) & \cite{sachdeva2014faster} 
        \\ \hline
        \(\exp(-(\gamma x)^2), \gamma \geq 0\) & \(t_{\gamma}\) & \(O\bigl(\sqrt{t_{\gamma} \ln(1/\delta)}\bigr)\) 
        & Eq.~\eqref{eq:gauss_approx}  & \(\delta \in (0,1/2]\) & \cite{low2017hamiltonian}
        \\ \hline
        \(\erf(\lambda x), \lambda > 0\) & \(t_{\lambda}\) 
        & \(O(\sqrt{t_{\lambda} \ln(1/\delta)})\) & Eq.~\eqref{eq:erf_approx} & \(\delta \in (0,O(1)]\) 
        & \cite{low2017hamiltonian}
        \\ \hline
        \(\cos^n y\), \(\sin^n y\), \(y \in \mathbb{R}\) & \(n\) 
        & \(O\bigl(\sqrt{n \ln(1/\delta)}\bigr)\) & Eqs.~\eqref{eq:mono_cos_approx} and \eqref{eq:mono_sin_approx}
        & \(\delta > 0\) & \multirow{4}{*}{This paper}
        \\ \cline{1-5}
        \(\exp(-\beta(1+\cos y))\), \(\exp(-\beta(1+\sin y))\) & \(t_{\beta}\) 
        & \(O\bigl(\sqrt{t_{\beta}\ln(1/\delta)}\bigr)\) & Eqs.~\eqref{eq:exp_cos_approx} and \eqref{eq:exp_sin_approx} 
        & \(\delta \in (0, 1/2]\) & 
        \\ \cline{1-5}
        \(\exp(-(\gamma \cos y)^2)\), \(\exp(-(\gamma \sin y)^2)\) & \(t_{\gamma}\) 
        & \(O\bigl(\sqrt{t_{\gamma} \ln(1/\delta)}\bigr)\) & Eqs.~\eqref{eq:gauss_cos_approx} and \eqref{eq:gauss_sin_approx}
        & \(\delta \in (0,1/2]\) & 
        \\ \cline{1-5}
        \(\erf(\lambda \cos y)\), \(\erf(\lambda \sin y)\) & \(t_{\lambda}\) 
        & \(O(\sqrt{t_{\lambda} \ln(1/\delta)})\) & Eqs.~\eqref{eq:erf_cos_approx} and \eqref{eq:erf_sin_approx} 
        & \(\delta \in (0,O(1)]\) & 
        \\ \hline
    \end{tabular}
    \caption{
        The list of polynomial-approximable functions and their approximations (Section~\ref{sec:poly_approx} and Appendix~\ref{appendix:sine_ver}).
        The second column shows either the degrees (\(n\)) of the monomials or the truncation degrees (\(t_\beta,t_\lambda\), and \(t_\gamma\)) of the exponential functions.
        The fifth column gives the approximation errors.
        The degree of a Laurent polynomial is defined as the highest power of \(z\).
    }
    \label{tab:poly_approx}
\end{table*}

In this section, we first review the Chebyshev polynomial approximation to monomials, exponential, Gaussian, and error functions \cite{sachdeva2014faster,low2017hamiltonian}.
For completeness, the theorems, corollaries, and their proofs in Refs.~\cite{sachdeva2014faster,low2017hamiltonian} are also given in Appendix~\ref{appendix:proof}.
We then derive our new Laurent polynomial approximation from the Chebyshev polynomial approximation.

A monomial \(x^n\) on \([-1,1]\) can be efficiently approximated by the following Chebyshev expansion \(p_{n,d}(x)\) with an error \(\delta\) \cite[Theorem 3.3]{sachdeva2014faster} (see Theorem~\ref{thm:monomial_approx}):
\begin{equation}
    x^n \approx p_{n,d}(x) \coloneqq \sum_{j=0}^{d} c_{n,j} T_j(x),
    \label{eq:mono_approx}
\end{equation}
where \(d = O\bigl(\sqrt{n \ln(1/\delta)}\bigr)\), and \(T_j(x) \coloneqq \cos(j \arccos x)\) denotes the Chebyshev polynomials of the first kind. 
Here, the coefficients \(c_{n,j}\) \cite{aulicino2022state,mason2002chebyshev} are defined as
\begin{equation}
    c_{n,j} =
    \begin{cases}
        \frac{1}{2^{n-1}} \binom{n}{(n-j)/2}, & \text{if } j \neq 0 \text{ and } n - j \text{ even}, \\
        \frac{1}{2^n} \binom{n}{n/2}, & \text{if } j = 0 \text{ and } n \text{ even}, \\
        0, & \text{otherwise,}
    \end{cases}
    \label{eq:mono_coeff}
\end{equation}
where the factor \(1/2\) is incorporated into \(c_{n,0}\), and \(c_{n,j} = 0\) for unmatched parities.
This efficient approximation arises from the equioscillation property of the Chebyshev polynomials \cite{sachdeva2014faster,mason2002chebyshev,borel1905leccons}, and
is used for approximating exponential, Gaussian, and error functions.

By approximating each monomial \(x^k\) in the Maclaurin series of \(\rme^{-\beta(1+x)}\) 
(\(\beta > 0\)) by \(p_{k,d}(x)\),
\(q_{\beta,t_{\beta},d}(x)\) can efficiently approximate \(\rme^{-\beta(1+x)}\) with an error \(\delta\):
\begin{equation}
    \rme^{-\beta(1+x)} \approx 
    q_{\beta,t_{\beta},d}(x) \coloneqq \rme^{-\beta} \sum_{k=0}^{t_{\beta}} \frac{(-\beta)^k}{k!} p_{k,d}(x),
    \label{eq:exp_approx}
\end{equation}
where the truncation degree is \(t_{\beta} = O(\beta + \ln(1/\delta))\) \cite{low2017hamiltonian}, and \(d = O\bigl( \sqrt{t_{\beta} \ln(1/\delta)} \bigr)\) \cite[Lemma 4.2]{sachdeva2014faster} (see Theorem~\ref{col:exp_approx}).

Since \(\rme^{-(\gamma x)^2} = \rme^{-\frac{\gamma^2}{2}(T_2(x) + 1)}\) (\(\gamma \geq 0\)),
we can substitute \(\beta \mapsto \gamma^2/2\) and \(x \mapsto T_2(x)\) in Eq.~\eqref{eq:exp_approx} \cite{low2017hamiltonian}.
Using the identity \(T_j(T_2(x)) = T_{2j}(x)\),
we obtain the following Chebyshev polynomial approximation to the Gaussian function with an error \(\delta\) \cite[Corollary 3]{low2017hamiltonian} (see Corollary~\ref{col:gauss_approx}):
\begin{equation}
    \rme^{-(\gamma x)^2} \approx 
    \rme^{-\gamma^2/2} \sum_{k=0}^{t_{\gamma}} \frac{(-\gamma^2/2)^k}{k!} 
    \sum_{j=0}^{d} c_{k,j} T_{2j}(x),
    \label{eq:gauss_approx}
\end{equation}
where \(t_{\gamma} = O(\gamma^2 + \ln(1/\delta))\), and \(d = O\bigl( \sqrt{t_{\gamma} \ln(1/\delta)}\bigr)\).

The error function is defined as: \(\erf(\lambda x) = (2\lambda/\sqrt{\pi})\int_{0}^{x} \rme^{-(\lambda u)^2} \mathrm{d}u\), where \(\lambda > 0\).
Applying Eq.~\eqref{eq:gauss_approx} to the definition of the error function yields the following \(\delta\)-approximation \cite[Corollary 4]{low2017hamiltonian} (see Corollary~\ref{col:erf_approx}):
\begin{equation}
    \erf(\lambda x)
    \approx \frac{2\lambda \rme^{-\frac{\lambda^2}{2}}}{\sqrt{\pi}} \sum_{k=0}^{t_{\lambda}} \frac{(-\lambda^2/2)^k}{k!} \int_{0}^{x} p_{k,d}(T_2(u)) \mathrm{d}u,
    \label{eq:erf_approx}
\end{equation}
where \(t_{\lambda} = O(\lambda^2 + \ln(1/\delta))\), and \(d = O\bigl( \sqrt{t_{\lambda} \ln(1/\delta)}\bigr)\).
The integral \(\int_{0}^{x} p_{k,d}(T_2(u)) \mathrm{d}u\) evaluates to
\begin{equation}
    \sum_{j=0}^{d} c_{k,j} \left(\frac{T_{2j+1}(x)}{2(2j+1)} - \frac{T_{|2j-1|}(x)}{2(2j-1)} \right),
    \label{eq:erf_approx_int}
\end{equation}
where \(T_{|-1|}(x) \coloneqq T_{1}(x)\).

Now, we present our Laurent polynomial approximation.
We extend the Chebyshev polynomial approximation to Laurent polynomial representations in the variable \(z \coloneqq \rme^{iy}\) for \(y \in \mathbb{R}\).
We only consider the Laurent polynomials of degree \(l\) of the form \(\sum_{k=-l}^{l} c_k z^k\), where \(c_{\pm l} \neq 0\).
Substituting \(x \mapsto \cos y\) into Eq.~\eqref{eq:mono_approx} yields a Laurent polynomial of degree \(d = O(\sqrt{n \ln(1/\delta)})\), which is a \(\delta\)-approximation to \(\cos^n y\):
\begin{equation}
    \cos^n y \approx p_{n,d}(\cos y) = \sum_{j=0}^{d} \frac{c_{n,j}}{2} (z^j + z^{-j}).
    \label{eq:mono_cos_approx}
\end{equation}
Eq.~\eqref{eq:mono_cos_approx} is a useful approximation as \(\cos^n y\) serves as a window function for filtering out unnecessary signals.
Note that our derivation provides a more practical degree \(d\) compared to Ge et al.'s~\cite{ge2019faster}.
They expressed the order \(d\) in terms of parameters usually unknown in advance, including a Hamiltonian's spectral gap.

We use Eq.~\eqref{eq:mono_cos_approx} for approximating the following various functions.
The Laurent polynomial \(p_{n,d}(\cos y)\) can be used for approximating \(\rme^{-\beta(1+\cos y)}\) with an error \(\delta\) as follows:
\begin{equation}
    \rme^{-\beta(1+\cos y)} \approx \rme^{-\beta} \sum_{k=0}^{t_{\beta}} \frac{(-\beta)^k}{k!} p_{k,d}(\cos y),
    \label{eq:exp_cos_approx}
\end{equation}
where \(t_{\beta} = O(\beta + \ln(1/\delta))\), and \(d = O\bigl( \sqrt{t_{\beta} \ln(1/\delta)} \bigr)\).

Similarly, a \(\delta\)-approximation to \(\rme^{-(\gamma \cos y)^2}\) is
\begin{equation}
    \rme^{-(\gamma \cos y)^2} \approx \rme^{-\gamma^2/2} \sum_{k=0}^{t_{\gamma}} \frac{(-\gamma^2/2)^k}{k!} \sum_{j=0}^{d} \frac{c_{k,j}}{2} (z^{2j} + z^{-2j}),
    \label{eq:gauss_cos_approx}
\end{equation}
where \(t_{\gamma} = O(\gamma^2 + \ln(1/\delta))\), and \(d = O\bigl( \sqrt{t_{\gamma} \ln(1/\delta)} \bigr)\).

We also obtain a \(\delta\)-approximation to \(\erf(\lambda \cos y)\):
\begin{equation}
    \frac{2\lambda \rme^{-(\lambda/2)^2}}{\sqrt{\pi}} \sum_{k=0}^{t_{\lambda}} \frac{(-(\lambda/2)^2)^k}{k!} \int_{0}^{\cos y} p_{k,d}(T_2(u)) \mathrm{d}u,
    \label{eq:erf_cos_approx}
\end{equation}
with \(t_{\lambda} = O(\lambda^2 + \ln(1/\delta))\), \(d = O\bigl( \sqrt{t_{\lambda} \ln(1/\delta)}\bigr)\), 
and the integral \(\int_{0}^{\cos y} p_{k,d}(T_2(u)) \mathrm{d}u\) simplifies to
\begin{equation}
    \sum_{j=0}^{d} \frac{c_{k,j}}{2} \left(\frac{z^{2j+1} + z^{-2j-1}}{2(2j+1)} - \frac{z^{|2j-1|} + z^{-|2j-1|}}{2(2j-1)} \right).
    \label{eq:erf_cos_approx_int}
\end{equation}

The approximations to \(\sin^n y\), \(\rme^{-\beta (1+\sin y)}\), \(\rme^{-(\gamma \sin y)^2}\), and \(\erf(\lambda \sin y)\) can be directly obtained using the identity \(\sin y = \cos (\pi/2 - y)\). Their explicit formulas are provided in Appendix~\ref{appendix:sine_ver}.
We collectively refer to the approximations to the functions of cosine and sine as Laurent polynomial approximations.
Table~\ref{tab:poly_approx} summarizes this section, providing the full list of the polynomial-approximable functions in this paper.

\section{Efficient Approximations to a Linear Combination and a Product of Polynomial-approximable Functions}
\label{sec:lc_prod}
In Section~\ref{sec:poly_approx}, we reviewed and introduced the polynomial approximations.
However, they are limited to the specific functions listed in Table~\ref{tab:poly_approx}.
In this section, we present the main approximation theorems of this work: Theorems~\ref{thm:linear_comb} and \ref{thm:product}.
By constructing linear combinations or products of polynomial-approximable functions, we dramatically extend the scope of the polynomial approximations presented in Section~\ref{sec:poly_approx}.

We begin by introducing the necessary notations used in Theorems~\ref{thm:linear_comb} and \ref{thm:product}.
\(\mathcal{F}_{\mathrm{mon}} \coloneqq \left\{ x^n \mid n \in \mathbb{N} \right\}\) stands for the set of all monomials.
Similarly, we define the following function classes:
\(\mathcal{F}_{\mathrm{exp}} \coloneqq \{e^{-\beta(1+x)} \mid \beta > 0\}\),
\(\mathcal{F}_{\mathrm{gauss}} \coloneqq \{e^{-(\gamma x)^2} \mid \gamma \geq 0\}\), and
\(\mathcal{F}_{\mathrm{erf}} \coloneqq \{\erf(\lambda x) \mid \lambda >0\}\).
We use \(\mathcal{F}_{i}^{\mathrm{(cos)}}\) and \(\mathcal{F}_{i}^{\mathrm{(sin)}}\) to denote the cosine and sine versions of \(\mathcal{F}_i\), respectively.
For example, \(\mathcal{F}_{\mathrm{mon}}^{\mathrm{(cos)}} \coloneqq \left\{ \cos^n y \mid n \in \mathbb{N} \right\}\).
\(\mathcal{F}_{\mathrm{all}} \) denotes the set of all functions for which Chebyshev polynomial approximations on the interval \([-1,1]\) are viable:
\begin{equation}
    \mathcal{F}_{\mathrm{all}} := \mathcal{F}_{\mathrm{mon}} \cup \mathcal{F}_{\mathrm{exp}} \cup \mathcal{F}_{\mathrm{gauss}} \cup \mathcal{F}_{\mathrm{erf}}.
    \label{eq:f_all}
\end{equation}
We write \(\mathcal{G}_{\mathrm{all}}\) for the set of all functions for which Laurent polynomial approximations are possible:
\begin{equation}
    \mathcal{G}_{\mathrm{all}} := \bigcup_{i \in \{\mathrm{mon},\, \mathrm{exp},\, \mathrm{gauss},\, \mathrm{erf} \}} \left( \mathcal{F}_i^{\mathrm{(cos)}} \cup \mathcal{F}_i^{\mathrm{(sin)}} \right).
    \label{eq:g_all}
\end{equation}
The elements of the sets \(\mathcal{F}_{\mathrm{all}}\) and \(\mathcal{G}_{\mathrm{all}}\) are the polynomial-approximable functions.
Note that each function \(g_j(y) \in \mathcal{G}_{\mathrm{all}}\) is represented by a Laurent polynomial in the variable \(z\).
For a function formed as a sum or a product of a degree-\(d_1\) polynomial and a degree-\(d_2\) Laurent polynomial, 
we define the function's degree as the tuple: \((d_1,d_2)\).

We now present our two approximation theorems: one for a linear combination (Theorem~\ref{thm:linear_comb}) and the other for a product (Theorem~\ref{thm:product}) of the polynomial-approximable functions.
Their proofs are provided in Appendix~\ref{appendix:lc_prod_proof}.
\begin{theorem}
    \label{thm:linear_comb}
    Let \(F: [-1,1] \times \mathbb{R} \rightarrow \mathbb{R}\) be a linear combination of elements 
    \(f_i \in \mathcal{F}_{\mathrm{all}}\) and \(g_j \in \mathcal{G}_{\mathrm{all}}\), i.e.,
    \begin{equation}
        F(x,y) = \sum_{i=1}^{\mathcal{N}} a_i f_i(x) + \sum_{j=1}^{\mathcal{M}} b_j g_j(y),
    \end{equation}
    for real coefficients \(a_i, b_j\), where \(\mathcal{N} \geq 0\) and \(\mathcal{M} \geq 0\) represent the constant number of functions.
    Let \(h_j(z)\) denote the Laurent polynomial representation of each \(g_j(y)\).
    The degree of \(F\) is a tuple \((n_l,m_l)\), where
    \begin{equation}
        n_l \coloneqq \max_{1 \leq i \leq \mathcal{N}} \deg(f_i), \quad m_l \coloneqq \max_{1 \leq j \leq \mathcal{M}} \deg(h_j).
        \label{eq:thm_lc_degree}
    \end{equation}
    If the coefficients \(a_i\) and \(b_j\) satisfy the following sufficient condition for a constant \(\mathcal{C}\):
    \begin{equation}
        \sum_{i=1}^{\mathcal{N}} |a_i| + \sum_{j=1}^{\mathcal{M}} |b_j| = \mathcal{C},
        \label{eq:approx_condition}
    \end{equation}
    then, for any \(\delta \in (0, O(1)]\),
    the function \(F\) can be \(\delta\)-approximated by a polynomial of degree
    \(\bigl(O\bigl(\sqrt{n_l \ln(1/\delta)}\bigr), O\bigl(\sqrt{m_l \ln(1/\delta)}\bigr)\bigr)\). 
\end{theorem}

\begin{theorem}
    \label{thm:product}
    Let \(G: [-1,1] \times \mathbb{R} \rightarrow \mathbb{R}\) be a product of elements
    \(f_i \in \mathcal{F}_{\mathrm{all}}\) and \(g_j \in \mathcal{G}_{\mathrm{all}}\), i.e.,
    \begin{equation}
        G(x,y) = \prod_{i=1}^{\mathcal{N}} f_i(x) \prod_{j=1}^{\mathcal{M}} g_j(y),
    \end{equation}
    where \(\mathcal{N} \geq 0\) and \(\mathcal{M} \geq 0\) stand for the constant number of functions.
    Let \(h_j(z)\) denote the Laurent polynomial representation of each \(g_j(y)\).
    The degree of \(G\) is the tuple \((n_p,m_p)\), given by the sums 
    \begin{equation}
        n_p \coloneqq \sum_{i=1}^{\mathcal{N}} \deg(f_i), \quad m_p \coloneqq \sum_{j=1}^{\mathcal{M}} \deg(h_j).
        \label{eq:thm_prod_degree}
    \end{equation}
    Then, for any \(\delta \in (0, O(1)]\), 
    the function \(G\) can be \(\delta\)-approximated by a polynomial of degree
    \(\bigl(O\bigl(\sqrt{n_p \ln(1/\delta)}\bigr), O\bigl(\sqrt{m_p \ln(1/\delta)}\bigr)\bigr)\).
\end{theorem}
\begin{remark}
    In Theorems~\ref{thm:linear_comb} and \ref{thm:product}, \(\mathcal{N}\) and \(\mathcal{M}\) are nonnegative integers, allowing the theorems to cover single-variable cases: \(F(x)\), \(F(y)\), \(G(x)\), and \(G(y)\).
    We adopt the conventions \(\max_{1 \leq i \leq 0} (\cdot) = 0\), \(\prod_{i=1}^{0} (\cdot) = 1\), and \(\sum_{i=1}^{0} (\cdot) = 0\) to ensure the theorems hold for the cases when \(\mathcal{N} = 0\) or \(\mathcal{M} = 0\).     
    Eq.~\eqref{eq:approx_condition} is a sufficient condition for the linear combination polynomial approximation.
    If the sum in Eq.~\eqref{eq:approx_condition} grows exponentially with the function's original degree, the approximation requires a higher-degree polynomial. 
\end{remark}

\section{Synthesizing Quadratically Shallow Circuits for the Polynomial-Approximable Functions}
\label{sec:synthesizing}
\begin{table*}[t!]
    \begin{tabular}{|c|c||c|c|}
        \hline
        Hamiltonian Function & Circuit Depth & Circuit Depth (Poly. Approx.) & Error
        \\ \hline
        \(\ham^n\) & \(O(n D_{\text{B}})\) & \(O(\sqrt{n \ln(1/\delta)} D_{\text{B}})\) & \(\delta\)
        \\ \hline
        \(\exp(-\beta(I + \ham))\) & \(O(t_{\beta} D_{\text{B}})\) & \(O(\sqrt{t_{\beta} \ln(1/\delta)} D_{\text{B}})\) & \(\delta\)
        \\ \hline
        \(\exp(-(\gamma \ham)^2)\) & \(O(t_{\gamma} D_{\text{B}})\) & \(O(\sqrt{t_{\gamma} \ln(1/\delta)} D_{\text{B}})\) & \(\delta\)
        \\ \hline
        \(\erf(\lambda \ham)\) & \(O(t_{\lambda} D_{\text{B}})\) & \(O(\sqrt{t_{\lambda} \ln(1/\delta)} D_{\text{B}})\) & \(\delta\) 
        \\ \hline
        \(\cos^n (\ham)\), \(\sin^n (\ham)\)
        & \(O\bigl(n^{\frac{1+v}{v}} D_{\text{ST}}\bigr)\) 
        & \(O\bigl(\bigl(\sqrt{n \ln(1/\delta)}\bigr)^{\frac{1+v}{v}} D_{\text{ST}}\bigr)\) & \(O(\delta)\)
        \\ \hline
        \(\exp(-\beta(I + \cos \ham))\), \(\exp(-\beta(I + \sin \ham))\) 
        & \(O\bigl({t_{\beta}}^{\frac{1+v}{v}} D_{\text{ST}}\bigr)\) 
        & \(O\bigl(\bigl(\sqrt{t_{\beta} \ln(1/\delta)}\bigr)^{\frac{1+v}{v}} D_{\text{ST}}\bigr)\) & \(O(\delta)\)
        \\ \hline
        \(\exp(-(\gamma \cos \ham)^2)\), \(\exp(-(\gamma \sin \ham)^2)\) 
        & \(O\bigl({t_{\gamma}}^{\frac{1+v}{v}} D_{\text{ST}}\bigr)\) 
        & \(O\bigl(\bigl(\sqrt{t_{\gamma} \ln(1/\delta)}\bigr)^{\frac{1+v}{v}} D_{\text{ST}}\bigr)\) & \(O(\delta)\)
        \\ \hline
        \(\erf(\lambda\cos \ham)\), \(\erf(\lambda\sin \ham)\) 
        & \(O\bigl({t_{\lambda}}^{\frac{1+v}{v}} D_{\text{ST}}\bigr)\) 
        & \(O\bigl(\bigl(\sqrt{t_{\lambda} \ln(1/\delta)}\bigr)^{\frac{1+v}{v}} D_{\text{ST}}\bigr)\) & \(O(\delta)\)
        \\ \hline
        \(F(\ham)\) (Eq.~\eqref{eq:ham_linear}) 
        & \(O\bigl(n_l D_{\text{B}} + m_l^{\frac{1+v}{v}} D_{\text{ST}}\bigr)\)
        & \(O\bigl(\sqrt{n_l \ln(1/\delta)} D_{\text{B}} + \bigl(\sqrt{m_l \ln(1/\delta)}\bigr)^{\frac{1+v}{v}} D_{\text{ST}}\bigr)\) & \(O(\delta)\)
        \\ \hline
        \(G(\ham)\) (Eq.~\eqref{eq:ham_prod}) 
        & \(O\bigl(n_p D_{\text{B}} + m_p^{\frac{1+v}{v}} D_{\text{ST}}\bigr)\)
        & \(O\bigl(\sqrt{n_p \ln(1/\delta)} D_{\text{B}} + \bigl(\sqrt{m_p \ln(1/\delta)}\bigr)^{\frac{1+v}{v}} D_{\text{ST}}\bigr)\) & \(O(\delta)\)
        \\ \hline
    \end{tabular}
    \caption{
        Summary of quadratic reductions in circuit depths for synthesizing polynomial-approximable functions of \(\ham\) using the methods in Sections~\ref{sec:poly_approx} and~\ref{sec:lc_prod}.
        \(D_{\text{B}}\) and \(D_{\text{ST}}\) follow their definitions in the main text.
        The rightmost column represents the error of the approximated circuit, whose scaling is not worsened by the approximation, except \(\ham^n\).
    }
    \label{tab:hamil_function}
\end{table*}

The Chebyshev and Laurent polynomial approximations in Section~\ref{sec:poly_approx}, together with Theorems~\ref{thm:linear_comb} and \ref{thm:product}, enable the synthesis of the polynomial-approximable functions of \(\ham\) with quadratically shallow quantum circuits.
The \(N\)-qubit Hamiltonian \(\ham\) is assumed to be \(k\)-local and given as a linear combination of \(J\) Pauli strings:
\begin{equation}
    \ham = \sum_{\ell=0}^{J-1} \frac{\kappa_{\ell}}{\alpha} P_{\ell},
    \label{eq:ham}
\end{equation}
where \(\ham\) is normalized by \(\alpha = \sum_{\ell=0}^{J-1} |\kappa_{\ell}|\) to ensure \(\lVert \ham \rVert \leq 1\), and each \(P_{\ell}\) is a tensor product of at most \(k = O(1)\) Pauli operators.

We employ QSP and generalized quantum signal processing (GQSP) \cite{de2022fourier,motlagh2024generalized,haah2019product} to synthesize the polynomials on quantum circuits.
QSP can implement a polynomial \(P(\ham)\) of degree \(d_1\) using \(O(d_1)\) block-encodings \cite{gilyen2019quantum}.
A QSP sequence involves \(O(d_1)\) phase factors, which can be computed in time \(O(d_1^2)\) \cite{dong2021efficient}.
GQSP can implement a Laurent polynomial \(L(U)\) of degree \(d_2\) using \(O(d_2)\) 0-controlled \(U\) and 1-controlled \(U^{\dagger}\) operations, where \(U \coloneqq e^{\mathrm{i}\ham}\) \cite{berntson2025complementary,motlagh2024generalized}.
The \(O(d_2)\) phase factors in a GQSP circuit can be computed in time \(O(d_2 \log_2 d_2)\) \cite{motlagh2024generalized}.
We place the formal theorems and reviews of QSP and GQSP in Appendices~\ref{appendix:review_qsp} and~\ref{appendix:review_gqsp}, respectively, as this section focuses on circuit depth reduction achieved by the approximations.

For the Hamiltonian functions \(P(\ham)\), their QSP circuits require \(O(d_1)\) block-encodings of \(\ham\).
To construct a block-encoding of \(\ham\), \(O(\log_2 J)\) ancilla qubits are needed to encode the \(J\) control states associated with each Pauli string \(P_{\ell}\), and the entire block-encoding can be implemented with a CNOT depth of \(D_{\text{B}} \coloneqq O(J(k+\log_2 J))\) \cite{berry2015simulating,shende2005synthesis,da2022linear}.
Hence, the total QSP circuit depth is \(O(d_1 D_{\text{B}})\).

GQSP circuits for synthesizing the Laurent polynomials involve both controlled \(U\) and \(U^{\dagger}\).
In this work, we consider implementing the controlled Hamiltonian simulations using the \(2v\)th-order symmetric Suzuki-Trotter decomposition.
A Laurent polynomial \(L\) of degree \(d_2\) accumulates \(O(d_2^2)\) additive and multiplicative Trotter errors.
To ensure that the total Trotter error for implementing the \(d_2\)-degree Laurent polynomial is bounded by \(O(\delta)\),
the resulting GQSP circuit depth for a degree-\(d_2\) Laurent polynomial is \(O(d_2^{1+1/v} D_{\text{ST}})\), where we define \(D_{\text{ST}} \coloneqq 5^{v-1} Jk/\delta^{1/(2v)}\) (see Appendix~\ref{appendix:analysis_gqsp_circ} for detail).

The Chebyshev and Laurent polynomial approximations yield quadratically shallow circuits: \(O(\sqrt{d_1 \ln(1/\delta)} D_{\text{B}})\) and \(O((\sqrt{d_2 \ln(1/\delta)})^{1+1/v} D_{\text{ST}})\), respectively.
Furthermore, the quadratic reduction in the number of phase factors enhances numerical stability and reduces the classical computational cost of the phase-factor finding algorithms.

Our theorems can be used to approximate any Hamiltonian function that is either a linear combination or a product of the polynomial-approximable functions, i.e.,
\begin{equation}
    F(\ham) = \sum_{i=1}^{\mathcal{N}} a_i f_i(\ham) + \sum_{j=1}^{\mathcal{M}} b_j g_j(\ham),
    \label{eq:ham_linear}
\end{equation}
and
\begin{equation}
    G(\ham) = \prod_{i=1}^{\mathcal{N}} f_i(\ham) \prod_{j=1}^{\mathcal{M}} g_j(\ham),
    \label{eq:ham_prod}
\end{equation}
where we use the shorthand \(F(\ham) \coloneqq F(\ham, \ham)\) and \(G(\ham) \coloneqq G(\ham, \ham)\) for brevity.
Note that Eqs.~\eqref{eq:ham_linear} and~\eqref{eq:ham_prod} follow the notations given in Theorems~\ref{thm:linear_comb} and~\ref{thm:product}, respectively, and include the cases when \(\mathcal{N} = 0\) or \(\mathcal{M} = 0\).
For the linear combination \(F(\ham)\), Theorem~\ref{thm:linear_comb} reduces the required circuit depth of \(O\bigl(n_l D_{\text{B}} + m_l^{1+1/v} D_{\text{ST}}\bigr)\) to \(O\bigl(\sqrt{n_l \ln(1/\delta)} D_{\text{B}} + \bigl(\sqrt{m_l \ln(1/\delta)}\bigr)^{\frac{1+v}{v}} D_{\text{ST}}\bigr)\).
Theorem~\ref{thm:product} provides a similar improvement for the product function \(G(\ham)\) as shown in Table~\ref{tab:hamil_function}.
For example, the degree of \(\ham^n \rme^{-\ham^2 / \sigma^2}\) \cite{zhang2024measurement} can be quadratically reduced.
A special case of the Gaussian-power function is the Gaussian derivative filter \(\ham \rme^{-\ham^2 / \sigma^2}\) used in a filtered Krylov method \cite{wang2023quantum}.

The results of this section are summarized in Table~\ref{tab:hamil_function}.
Previous research suggested the Chebyshev polynomial approximation to \(\rme^{-\beta(I+\ham)}\) for imaginary time evolution \cite{gilyen2019quantum}, and eigenvalue filtering methods based on \(\rme^{-(\gamma \ham)^2}\) \cite{low2017hamiltonian} and \(\erf(\lambda H)\) \cite{low2017hamiltonian}.
Our Laurent polynomial approximation enables the efficient implementation of \(\cos^n(\ham)\), \(\rme^{-\beta(I+\cos \ham)}\), \(\rme^{-(\gamma \cos \ham)^2}\), and \(\erf(\lambda\cos \ham)\), as well as their sine formulas, for eigenvalue filtering.
Moreover, Theorems~\ref{thm:linear_comb} and~\ref{thm:product} further generalize the range of feasible filtering functions.
We also point out that approximating \(\ham^n\) for estimating Hamiltonian moments \(\left<\ham^n\right>\) may improve the computational complexity of the quantum power method \cite{seki2021quantum}.

\section{Conclusion}
\label{sec:conclusion}
We have proposed a general function approximation framework, stated in Theorems~\ref{thm:linear_comb} and \ref{thm:product}, whose generality is further enhanced by adopting our new Laurent polynomial approximation.
Our approach quadratically reduces circuit depths for many useful Hamiltonian functions (Table~\ref{tab:hamil_function}) used in quantum ground-state preparation and energy estimation algorithms.
This quadratic reduction in circuit depth may quadratically decrease the computational complexity of quantum algorithms that require the construction of Hamiltonian functions.
Moreover, our approximation is useful for preparing an initial state with a high overlap with the ground state, which typically needs a high-degree polynomial filter.

Theorem~\ref{thm:linear_comb} and \ref{thm:product} allow the approximation of generalized classical window functions \cite{prabhu2014window}, which are applicable for quantum eigenvalue and eigenstate filtering.
For instance, a wide range of target functions can be numerically fit to linear combinations of exponential functions (i.e., \(\sum_{k} c_k \rme^{- \beta_k (I+\ham)}\))
or Gaussian functions (i.e., \(\sum_{k} c_k \rme^{-(\gamma_k \ham)^2}\)).
A linear combination of power-of-sine and power-of-cosine window functions can be suggested: \(\sum_{n,m} c_n \cos^{n} \ham + b_m \sin^m \ham\).

Furthermore, a straightforward extension of our theorems can be achieved through a change of variables. 
By substituting \(x \mapsto x^n\), the generalized normal window \(\exp(-(\ham^2/\sigma^2)^n)\) truncated at degree \(\tau\) can be approximated by a polynomial of degree \(O\bigl(n\sqrt{\tau \ln(1/\delta)}\bigr)\).
We expect the function approximation in this work to serve as an essential step for efficient quantum ground-state preparation and energy estimation.

\begin{acknowledgments}
    The authors gratefully acknowledge the many helpful discussions with Gwonhak Lee during the preparation of this paper.
    This work was partly supported by Basic Science Research Program through the National Research Foundation of Korea (NRF), funded by the Ministry of Education, Science and Technology (NRF-2022M3H3A106307411, RS-2023-NR119931, RS-2025-03532992, RS-2025-07882969). This work was also partly supported by Institute for Information \& communications Technology Promotion (IITP) grant funded by the Korea government (MSIP) (No. 2019-0-00003, Research and Development of Core technologies for Programming, Running, Implementing and Validating of Fault-Tolerant Quantum Computing System). The Ministry of Trade, Industry, and Energy (MOTIE), Korea, also partly supported this research under the Industrial Innovation Infrastructure Development Project (Project No. RS-2024-00466693). JH is supported by the Yonsei University Research Fund of 2025-22-0140.
\end{acknowledgments}

\bibliography{references}

\onecolumngrid
\section*{Appendix}
\appendix

\section{The Chebyshev Polynomial Approximation Theorems and Proofs}
\label{appendix:proof}
For completeness, we explicitly state and prove the approximation theorem and the corollaries referenced in Section~\ref{sec:poly_approx}, which were developed in Refs.~\cite{low2017hamiltonian,sachdeva2014faster}.
Before proving the polynomial approximation theorems, we review Chapter 3 in Ref. \cite{sachdeva2014faster} to discuss the relevant properties of Chebyshev polynomials.
The recurrence relation for the Chebyshev polynomials is 
\begin{equation}
    T_n(x) = 2x T_{n-1}(x) - T_{n-2}(x), \quad n = 2, 3, \ldots,
    \label{eq:cheby_def}
\end{equation}
where
\begin{equation}
    T_0(x) = 1, \quad T_1(x) = x.
\end{equation}
By defining \(T_d(x) = T_{|d|}(x)\) for any \(d \in \mathbb{Z}^{-}\), Eq.~\eqref{eq:cheby_def} can be rearranged as follows \cite{sachdeva2014faster}:
\begin{equation}
    x \cdot T_d(x) = \frac{1}{2} \left(T_{d-1}(x) + T_{d+1}(x) \right).
    \label{eq:rand_cheby}
\end{equation}
Let \(X\) be a random variable taking values \(1\) and \(-1\), each with probability \(1/2\). Then Eq.~\eqref{eq:rand_cheby} can be expressed as 
\begin{equation}
    x \cdot T_d(x) = \mathbb{E}_X \left[ T_{d+X}(x) \right].
\end{equation}
This observation extends to the Chebyshev expansion of the monomial \(x^n\).
Let \(X_1, \ldots, X_n\) be independent and identically distributed random variables such that \(\Pr(X_i = 1) = \Pr(X_i = -1) = 1/2\) for all \(i\).
Define \(Z_n \coloneqq \sum_{i=1}^{n} X_i\) and \(Z_0 \coloneqq 0\).
The Chebyshev expansion of \(x^n\) can be formulated as the expectation of these random variables as the following theorem \cite[Theorem 3.1]{sachdeva2014faster}:
\begin{theorem}
    \label{thm:monomial_exact}
    Let \(n \in \mathbb{N}_{0}\). For the random variables \(X_1,\ldots,X_n\) defined above, the following Chebyshev series expands the monomial \(x^n\):
    \begin{equation}
        x^n = \underset{X_1,\ldots,X_n}{\mathbb{E}} \left[ T_{Z_n}(x) \right],
        \label{eq:mono_prob}
    \end{equation}
    where \(T_{Z_n} \coloneqq T_{|Z_n|}\) for negative \(Z_n\).
\end{theorem}
\begin{proof}
    The proof is by induction on \(n\). The base case of \(n = 0\) holds since \(x^0 = \mathbb{E} \left[ T_{Z_0}(x) \right] = T_0(x)\).
    Next, we assume Eq.~\eqref{eq:mono_prob} holds for \(n \geq 0\). Then, for \(n+1\), we have:
    \begin{equation}
    \begin{split}
        x^{n+1}
        & = x \cdot \underset{X_1,\ldots,X_n}{\mathbb{E}} \left[ T_{Z_n}(x) \right] \\
        & = \underset{X_1,\ldots,X_n}{\mathbb{E}} \left[ x \cdot T_{Z_n}(x) \right] \\
        & = \underset{X_1,\ldots,X_n}{\mathbb{E}} \left[ \frac{1}{2} \left(T_{Z_n-1}(x) + T_{Z_n+1}(x) \right) \right] \\
        & = \Pr(X_{n+1} = -1) \cdot \underset{X_1,\ldots,X_n}{\mathbb{E}} \left[ T_{Z_n-1}(x) \right]
        + \Pr(X_{n+1} = 1) \cdot \underset{X_1,\ldots,X_n}{\mathbb{E}} \left[ T_{Z_n+1}(x) \right] \\
        & = \underset{X_1,\ldots,X_{n+1}}{\mathbb{E}} \left[ T_{Z_{n+1}}(x) \right],
    \end{split}        
    \end{equation}
    where the third equality follows from Eq.~\eqref{eq:rand_cheby}.
\end{proof}
The explicit formula of Eq.~\eqref{eq:mono_prob}, using the definition \(T_d(x) = T_{|d|}(x)\), is given by
\begin{equation}
    x^n = \sum_{k=0}^{n} c_{n,k} T_k(x),
    \label{eq:mono_exact}
\end{equation}
where the coefficients \(c_{n,k}\) are given in Eq.~\eqref{eq:mono_coeff}.
Note that the probability distribution of the sum of the random variables \(X_1, \ldots, X_n\) has exponentially decreasing tails, as it forms a simple symmetric random walk centered at zero after \(n\) steps \cite{montgomery1990distribution}.
The Chernoff bound for the random variables implies that the probability of observing a large \(|Z_n|\) is exponentially small.
\begin{theorem}
    \label{thm:chernoff}
    (Chernoff bound \cite[Chapter 4.3, Corollary 4.8]{mitzenmacher2017probability}).
    Let \(X_1,\ldots,X_n\) be the random variables defined above.
    Then,
    \begin{equation}
        \Pr(|Z_n| \geq a) \leq 2 \rme^{-a^2/2n}, \quad \forall a \in \mathbb{R}^{+}.
    \end{equation}
\end{theorem}

We are now in a position to present the theorem for approximating monomials.
\begin{theorem}
    (Polynomial Approximation to monomials \cite[Theorem 3.3]{sachdeva2014faster}).
    \label{thm:monomial_approx}
    Let \(n, d \in \mathbb{N}\).
    The uniform norm of the error in approximating \(x^n\) by the degree-\(d\) polynomial \(p_{n,d}\) defined in Eq.~\eqref{eq:mono_approx} satisfies
    \begin{equation}
        \underset{x \in [-1,1]}{\sup} \left| x^n - p_{n,d}(x) \right| \leq 2 \rme^{-d^2/2n}.
    \end{equation}
\end{theorem}
\begin{proof}
    \begin{equation}
    \begin{split}
        \underset{x \in [-1,1]}{\sup} \left| x^n - p_{n,d}(x) \right| 
        & = \underset{x \in [-1,1]}{\sup} \left| \sum_{j=d+1}^{n} c_{n,j} T_j(x) \right| \\
        & \leq \underset{x \in [-1,1]}{\sup} \sum_{j=d+1}^{n} c_{n,j} \left| T_j(x) \right| \\
        & \leq \sum_{j=d+1}^{n} c_{n,j} \cdot \left( \underset{x \in [-1,1]}{\sup} \left| T_j(x) \right| \right) \\
        & \leq \sum_{j=d+1}^{n} c_{n,j} \cdot 1 \\
        & \leq \Pr(|Z_n| > d) \\
        & \leq 2\rme^{-d^2/2n},
    \end{split}
    \end{equation}
    where the first inequality follows as \(c_{n,j} \geq 0\) for all \(j\), and the last inequality follows from the Chernoff bound.
    To upper bound the error by \(\delta\), it suffices to choose \(d \geq \left\lceil \sqrt{2n \ln(2/\delta)} \right\rceil\).
\end{proof}

We now proceed to state a theorem, corollaries, and their proofs about polynomial approximations to \(\rme^{-\beta(1+x)}\), \(\rme^{-(\gamma x)^2}\), and \(\erf(\lambda(x-b))\).
\begin{theorem}
    (Polynomial approximation to \(\rme^{-\beta(1+x)}\) \cite[Lemma 4.2]{sachdeva2014faster}).
    \label{col:exp_approx}
    Let \(\beta >0\) and \(\delta \in (0,1/2]\).
    For \(t_{\beta} = O(\beta + \ln(1/\delta))\) and \(d = O\bigl(\sqrt{t_{\beta} \ln(1/\delta)}\bigr)\), 
    the degree-\(d\) polynomial \(q_{\beta,t_{\beta},d}(x)\) defined in Eq.~\eqref{eq:exp_approx} \(\delta\)-approximates \(\rme^{-\beta(1+x)}\) over the interval \([-1,1]\):
    \begin{equation}
        \underset{x \in [-1,1]}{\sup} \left| \rme^{-\beta(1+x)} - q_{\beta,t_{\beta},d}(x) \right| \leq \delta.
    \end{equation}
\end{theorem}
\begin{proof}
    The error can be separated into two terms:
    \begin{equation}
    \begin{split}
        & \underset{x \in [-1,1]}{\sup} \left| \rme^{-\beta(1+x)} - \rme^{-\beta} \sum_{k=0}^{t_{\beta}} \frac{(-\beta)^k}{k!} p_{k,d}(x) \right| \\ 
        & \leq
        \underset{x \in [-1,1]}{\sup} \left| \rme^{-\beta} \sum_{k=0}^{t_{\beta}} \frac{(-\beta)^k}{k!} \left(x^k - p_{k,d}(x)\right) \right| 
        + \underset{x \in [-1,1]}{\sup} \left| \rme^{-\beta} \sum_{k=t_{\beta}+1}^{\infty} \frac{(-\beta)^k}{k!} x^k \right|,
    \end{split}
    \end{equation}
    where the first term is the sum of the monomial approximation errors from the truncated Maclaurin series, and the second term is the truncation error itself.
    Each term is upper bounded by \(\delta/2\) to ensure the total error is at most \(\delta\).
    The second term can be upper bounded as follows:
    \begin{equation}
    \begin{split}
        \underset{x \in [-1,1]}{\sup} \left| \rme^{-\beta} \sum_{k=t_{\beta}+1}^{\infty} \frac{(-\beta)^k}{k!} x^k \right|
        & \leq \rme^{-\beta} \sum_{k=t_{\beta}+1}^{\infty} \frac{\beta^k}{k!} \\
        & \leq \rme^{-\beta} \sum_{k=t_{\beta}+1}^{\infty} \left(\frac{\beta \rme}{k}\right)^k \\
        & \leq \rme^{-\beta} \sum_{k=t_{\beta}+1}^{\infty} \rme^{-k} \\
        & \leq \rme^{-\beta - t_{\beta}} \leq \frac{\delta}{2},
    \end{split}
    \end{equation}
    where the second inequality follows from the lower bound of Stirling's approximation, \(k! \geq (k/\rme)^k\), and 
    we assume \(t_{\beta} \geq \beta \rme^2\) in the third inequality.
    Thus, it suffices to set \(t_{\beta} = \left\lceil \max(\beta \rme^2, \ln(2/\delta)) \right\rceil = O(\beta + \ln(1/\delta))\).

    Finally, the first term can be upper bounded as follows:
    \begin{equation}
    \begin{split}
        \underset{x \in [-1,1]}{\sup} \left| \rme^{-\beta} \sum_{k=0}^{t_{\beta}} \frac{(-\beta)^k}{k!} \left(x^k - p_{k,d}(x)\right) \right|
        & \leq \rme^{-\beta} \sum_{k=0}^{t_{\beta}} \frac{\beta^k}{k!} \cdot 2\rme^{-d^2/2k} \\
        & \leq 2\rme^{-d^2/2t_{\beta}} \cdot \rme^{-\beta} \sum_{k=0}^{t_{\beta}} \frac{\beta^k}{k!} \\
        & \leq 2\rme^{-d^2/2t_{\beta}} \cdot \rme^{-\beta} \sum_{k=0}^{\infty} \frac{\beta^k}{k!} \\
        & \leq 2\rme^{-d^2/2t_{\beta}} \leq \frac{\delta}{2}.
    \end{split}
    \end{equation}
    Therefore, it suffices to choose \(d = \left\lceil \sqrt{2t_{\beta} \ln(4/\delta)} \right\rceil = O\bigl(\sqrt{t_{\beta} \ln(1/\delta)}\bigr)\).
\end{proof}

\begin{corollary}
    (Polynomial approximation to \(\rme^{-(\gamma x)^2}\) \cite[Corollary 3]{low2017hamiltonian}).
    \label{col:gauss_approx}
    Let \(\gamma \geq 0\) and \(\delta \in (0,1/2]\).
    For \(t_{\gamma} = O(\gamma^2 + \ln(1/\delta))\) and \(d = O\bigl( \sqrt{t_{\gamma} \ln(1/\delta)} \bigr)\), 
    the \(2d\)-degree polynomial \(q_{\gamma^2/2,t_{\gamma},d}(T_2(x))\) satisfies 
    \begin{equation}
        \underset{x \in [-1,1]}{\sup} \left| \rme^{-(\gamma x)^2} - q_{\gamma^2/2,t_{\gamma},d}(T_2(x))\right| \leq \delta.
    \end{equation}
\end{corollary}
\begin{proof}
    This follows directly from Theorem~\ref{col:exp_approx}.
    Using the identity \(T_2(x) = 2x^2 - 1\), we can rewrite the error function as
    \begin{equation}
        \rme^{-(\gamma x)^2} = \rme^{-\frac{\gamma^2}{2}(T_2(x) + 1)}.
    \end{equation}
    Applying Theorem~\ref{col:exp_approx} with the substitutions \(\beta \mapsto \gamma^2/2\) and \(x \mapsto T_2(x)\),
    we obtain the following \(\delta\)-approximation to the Gaussian function:
    \begin{equation}
        q_{\gamma^2/2,t_{\gamma},d}(T_2(x)) = \rme^{-\gamma^2/2} \sum_{k=0}^{t_{\gamma}} \frac{(-\gamma^2/2)^k}{k!} p_{k,d}(T_2(x)),
        \label{eq:q_gamma}
    \end{equation}
    where \(t_{\gamma} = O(\gamma^2 + \ln(1/\delta))\), and \(d = O\bigl( \sqrt{t_{\gamma} \ln(1/\delta)}\bigr)\).
    We use the identity \(T_j(T_2(x)) = T_{2j}(x)\) to find the explicit formulas of \(p_{k,d}(T_2(x))\) as follows:
    \begin{equation}
        p_{k,d}(T_2(x)) = \sum_{j=0}^{d} c_{k,j} T_j(T_2(x)) = \sum_{j=0}^{d} c_{k,j} T_{2j}(x).
        \label{eq:tj2_t2j}
    \end{equation}
    Substituting Eq.~\eqref{eq:tj2_t2j} into the right hand side of Eq.~\eqref{eq:q_gamma}, we obtain
    \begin{equation}
        q_{\gamma^2/2, t_{\gamma}, d}(T_2(x)) = \rme^{-\gamma^2/2} \sum_{k=0}^{t_{\gamma}} \frac{(-\gamma^2/2)^k}{k!} \left( \sum_{j=0}^{d} c_{k,j} T_{2j}(x) \right).
    \end{equation}
\end{proof}

\begin{corollary}
    (Polynomial approximation to \(\erf(\lambda x)\) \cite[Corollary]{low2017hamiltonian}).
    \label{col:erf_approx}
    Let \(\lambda >0\), and \(\delta \in (0, O(1)]\).
    For \(t_{\lambda} = O(\lambda^2 + \ln(1/\delta))\) and \(d = O\bigl( \sqrt{t_{\lambda} \ln(1/\delta)} \bigr)\),
    the degree-\((2d+1)\) polynomial \(w_{\lambda, t_{\lambda}, d} \left(x\right)\), defined by
    \begin{equation}
        w_{\lambda, t_{\lambda}, d} (x) 
        = \frac{2\lambda \rme^{-\lambda^2/2}}{\sqrt{\pi}} \sum_{k=0}^{t_{\lambda}} \frac{(-\lambda^2/2)^k}{k!}
        \left(\sum_{j=0}^{d} c_{k,j} \left(\frac{T_{2j+1}(x)}{2(2j+1)} - \frac{T_{|2j-1|}(x)}{2(2j-1)} \right)\right),
    \end{equation}
    satisfies
    \begin{equation}
        \underset{x \in [-1,1]}{\sup} \left| \erf(\lambda x) - w_{\lambda, t_{\lambda}, d} (x) \right| \leq \delta.
    \end{equation}
\end{corollary}
\begin{proof}
    The error function is defined as
    \begin{equation}
        \erf(\lambda x) 
        \coloneqq \frac{2}{\pi} \int_{0}^{\lambda x} \rme^{-u^2} \mathrm{d}u 
        = \frac{2\lambda}{\sqrt{\pi}} \int_{0}^{x} \rme^{-(\lambda u)^2} \mathrm{d}u.
    \end{equation}
    Approximating the integrand \(\rme^{-(\lambda u)^2}\) with \(q_{\lambda^2/2,t_{\lambda},d}(T_2(u))\) in Corollary~\ref{col:gauss_approx}, we get
    \begin{equation}
        w_{\lambda, t_{\lambda}, d} \left(x\right) 
        = \frac{2\lambda}{\sqrt{\pi}} \int_{0}^{x} q_{\lambda^2/2,t_{\lambda},d}(T_2(u)) \mathrm{d}u
        = \frac{2\lambda \rme^{-\lambda^2/2}}{\sqrt{\pi}} \sum_{k=0}^{t_{\lambda}} \frac{(-\lambda^2/2)^k}{k!} 
        \int_{0}^{x} \left( \sum_{j=0}^{d} c_{k,j} T_{2j}(u) \right) \mathrm{d}u.
        \label{eq:w_err}
    \end{equation}
    Using the integral identity of the Chebyshev polynomials \cite{mason2002chebyshev} in the right hand side of Eq.~\eqref{eq:w_err},
    \begin{equation}
        \int_{0}^{x} T_j(u) du = \frac{T_{j+1}(x)}{2(j+1)} - \frac{T_{|j-1|}(x)}{2(j-1)} \quad (j \neq 1),
    \end{equation}
    we have the \(\delta\)-approximation to the error function \(\erf(\lambda x)\):
    \begin{equation}
        w_{\lambda, t_{\lambda}, d} \left(x\right) 
        = \frac{2\lambda \rme^{-\lambda^2/2}}{\sqrt{\pi}} \sum_{k=0}^{t_{\lambda}} \frac{(-\lambda^2/2)^k}{k!}
        \left(\sum_{j=0}^{d} c_{k,j} \left(\frac{T_{2j+1}(x)}{2(2j+1)} - \frac{T_{|2j-1|}(x)}{2(2j-1)} \right)\right).
    \end{equation}
\end{proof}

\section{The Laurent Polynomial Approximations to Sine Functions}
\label{appendix:sine_ver}
The Laurent polynomial approximations to \(\sin^n y\), \(\rme^{-\beta(1 + \sin y)}\), \(\rme^{-(\gamma \sin y)^2}\), and \(\erf(\lambda \sin y)\)
can be straightforwardly obtained from their cosine versions in Section~\ref{sec:poly_approx}, using the identity \(\sin y = \cos (\pi/2 - y)\).
For simplicity, we omitted the explicit formulas for these sine-based functions in the main text.
Here, we present them explicitly.

The Laurent polynomial approximation to \(\sin^n y\) with an error \(\delta\) is \(p_{n,d}(\sin y)\) as follows:
\begin{equation}
    \sin^n y \approx
    p_{n,d}(\sin y) = \sum_{j=0}^{d} \frac{c_{n,j}}{2} \cdot \mathrm{i}^j \left((-z)^{j} + z^{-j}\right),
    \label{eq:mono_sin_approx}
\end{equation}
where we have used \(\rme^{\mathrm{i}(\pi/2 - y)} = \mathrm{i}z^{-1}\) and \(\rme^{-\mathrm{i}(\pi/2 - y)} = -\mathrm{i}z\).

Using Eq.~\eqref{eq:mono_sin_approx}, the \(\delta\)-approximation to \(\rme^{-\beta(1 + \sin y)}\) is given by
\begin{equation}
    \rme^{-\beta (1+\sin y)} \approx \rme^{-\beta} \sum_{k=0}^{t_{\beta}} \frac{(-\beta)^k}{k!} p_{k,d}(\sin y),
    \label{eq:exp_sin_approx}
\end{equation}
where \(t_{\beta} = O(\beta + \ln(1/\delta))\), and \(d = O\bigl( \sqrt{t_{\beta} \ln(1/\delta)} \bigr)\).

For the Gaussian function \(\rme^{-(\gamma \sin y)^2}\), we have the following \(\delta\)-approximation:
\begin{equation}
    \rme^{-(\gamma \sin y)^2} \approx \rme^{-\gamma^2/2} \sum_{k=0}^{t_{\gamma}} \frac{(-\gamma^2/2)^k}{k!} p_{k,d}(T_2(\sin y)),
    \label{eq:gauss_sin_approx}
\end{equation}
where \(t_{\gamma} = O(\gamma^2 + \ln(1/\delta))\), and \(d = O\bigl( \sqrt{t_{\gamma} \ln(1/\delta)} \bigr)\).
Using the identity \(T_2(\sin y) = T_2(\cos(\pi/2 - y)) = \cos (2(\pi/2 - y))\), we can simplify \(p_{k,d}(T_2(\sin y))\) as
\begin{equation}
    p_{k,d}(T_2(\sin y)) = \sum_{j=0}^{d} \frac{c_{k,j}}{2} (-1)^j \left(z^{2j} + z^{-2j}\right).
    \label{eq:gauss_sin_approx_t2}
\end{equation}

The Laurent polynomial approximation to \(\erf(\lambda \sin y)\) with an error \(\delta\) is
\begin{equation}
    \erf(\lambda \sin y) \approx
    \frac{2\lambda \rme^{-(\lambda/2)^2}}{\sqrt{\pi}} \sum_{k=0}^{t_{\lambda}} \frac{(-(\lambda/2)^2)^k}{k!} \int_{0}^{\sin y} p_{k,d}(T_2(u)) \mathrm{d}u,
    \label{eq:erf_sin_approx}
\end{equation}
where \(t_{\lambda} = O(\lambda^2 + \ln(1/\delta))\), and \(d = O\bigl( \sqrt{t_{\lambda} \ln(1/\delta)} \bigr)\).
The integral \(\int_{0}^{\sin y} p_{k,d}(T_2(u)) \mathrm{d}u\) simplifies to
\begin{align}
    \frac{c_{k,0}}{2} (-\mathrm{i}) (z - z^{-1}) + \sum_{j=1}^{d} \frac{c_{k,j}}{2\mathrm{i}} (-1)^j
    \left(\frac{z^{2j+1} - z^{-2j-1}}{2(2j+1)} + \frac{z^{2j-1} - z^{-2j+1}}{2(2j-1)}\right).
    \label{eq:erf_sin_approx_int}
\end{align}

\section{Proofs to the Main Theorems}
We prove Theorems~\ref{thm:linear_comb} and~\ref{thm:product} in this section. \\
\label{appendix:lc_prod_proof}

\textbf{Theorem~\ref{thm:linear_comb}}
\emph{
    Let \(F: [-1,1] \times \mathbb{R} \rightarrow \mathbb{R}\) be a linear combination of elements 
    \(f_i \in \mathcal{F}_{\mathrm{all}}\) and \(g_j \in \mathcal{G}_{\mathrm{all}}\), i.e.,
    \begin{equation}
        F(x,y) = \sum_{i=1}^{\mathcal{N}} a_i f_i(x) + \sum_{j=1}^{\mathcal{M}} b_j g_j(y),
    \end{equation}
    for real coefficients \(a_i, b_j\), where \(\mathcal{N} \geq 0\) and \(\mathcal{M} \geq 0\) represent the constant number of functions.
    Let \(h_j(z)\) denote the Laurent polynomial representation of each \(g_j(y)\).
    The degree of \(F\) is a tuple \((n_l,m_l)\), where
    \begin{equation}
        n_l \coloneqq \max_{1 \leq i \leq \mathcal{N}} \deg(f_i), \quad m_l \coloneqq \max_{1 \leq j \leq \mathcal{M}} \deg(h_j).
    \end{equation}
    If the coefficients \(a_i\) and \(b_j\) satisfy the following sufficient condition for a constant \(\mathcal{C}\):
    \begin{equation}
        \sum_{i=1}^{\mathcal{N}} |a_i| + \sum_{j=1}^{\mathcal{M}} |b_j| = \mathcal{C},
        \label{eq:appendix_lc_assump}
    \end{equation}
    then, for any \(\delta \in (0, O(1)]\),
    the function \(F\) can be \(\delta\)-approximated by a polynomial of degree
    \(\bigl(O\bigl(\sqrt{n_l \ln(1/\delta)}\bigr), O\bigl(\sqrt{m_l \ln(1/\delta)}\bigr)\bigr)\). 
}
\begin{proof}
    \(h_j(z)\) is a Laurent polynomial representation of \(g_j(y)\) with a change of variables, i.e., \(g_j(y) = h_j(z)\).
    Then, \(F(x,y)\) can be rewritten as its Laurent polynomial representation \(F'(x,z)\) as follows:
    \begin{equation}
        F(x,y) = F'(x,z) \coloneqq \sum_{i=1}^{\mathcal{N}} a_i f_i(x) + \sum_{j=1}^{\mathcal{M}} b_j h_j(z).
    \end{equation}
    As \(f_i\) and \(h_j\) are polynomial-approximable functions,
    \(f_i\) and \(h_j\) can be \(\delta\)-approximated by polynomials of degree \(O\bigl( \sqrt{\deg(f_i) \ln(1/\delta)} \bigr)\) and 
    Laurent polynomials of degree \(O\bigl( \sqrt{\deg(h_j) \ln(1/\delta)} \bigr)\), respectively.

    The proof proceeds in two cases, depending on the value of \(\mathcal{C}\).
    First, we consider when \(0 \leq \mathcal{C} \leq 1\).
    Let \(\hat{f}_i\) and \(\hat{h}_j\) denote the \(\delta\)-approximations to \(f_i\) and \(h_j\) for all \(i\) and \(j\), respectively, so that
    \begin{equation}
        \sup_{x \in [-1,1]} |f_i(x) - \hat{f}_i(x)| \leq \delta, \quad
        \sup_{z \in \mathbb{T}} |h_j(z) - \hat{h}_j(z)| \leq \delta,
        \label{eq:appendix_lc_delta_approx}
    \end{equation}
    where \(\mathbb{T} \coloneqq \{z \in \mathbb{C}: |z| = 1\}\).
    Consider the linear combination of \(\hat{f}_i\) and \(\hat{h}_j\) as an approximation to \(F'(x,z)\), 
    \begin{equation}
        \hat{F}'(x,z) \coloneqq \sum_{i=1}^{\mathcal{N}} a_i \hat{f}_i(x) + \sum_{j=1}^{\mathcal{M}} b_{j} \hat{h}_j(z),
    \end{equation}
    where \(\deg(\hat{F}') = \bigl(O\bigl(\sqrt{n_l \ln(1/\delta)}\bigr), O\bigl(\sqrt{m_l \ln(1/\delta)}\bigr)\bigr)\).
    Then, the approximation error \(|F' - \hat{F}'|\) is upper bounded by \(\delta\), as shown below:
    \begin{equation}
    \begin{split}
        \sup_{x \in [-1,1], z \in \mathbb{T}} |F'(x,z) - \hat{F}'(x,z)| 
        & \leq \sup_{x \in [-1,1]} \left| \sum_{i=1}^{\mathcal{N}} a_i (f_i(x) - \hat{f}_i(x)) \right| 
        + \sup_{z \in \mathbb{T}} \left| \sum_{j=1}^{\mathcal{M}} b_j (h_j(z) - \hat{h}_j(z)) \right| \\
        & \leq \sum_{i=1}^{\mathcal{N}} |a_i| \cdot \sup_{x \in [-1,1]} |f_i(x) - \hat{f}_i(x)| 
        + \sum_{j=1}^{\mathcal{M}} |b_j| \cdot \sup_{z \in \mathbb{T}} |h_j(z) - \hat{h}_j(z)| \\
        & \leq \delta \cdot \left(\sum_{i=1}^{\mathcal{N}} |a_i| + \sum_{j=1}^{\mathcal{M}} |b_j|\right) \\
        & \leq \delta,
    \end{split}
    \end{equation}
    where the third inequality follows from the \(\delta\)-approximations in Eq.~\eqref{eq:appendix_lc_delta_approx},
    and the last inequality follows from the assumption in Eq.~\eqref{eq:appendix_lc_assump}.

    When \(\mathcal{C} > 1\), we require smaller approximation errors for \(f_i\) and \(h_j\).
    For all \(i\) and \(j\), the approximation errors are upper bounded by \(\delta/\mathcal{C}\):
    \begin{equation}
        \sup_{x \in [-1,1]} |f_i(x) - \hat{f}_i(x)| \leq \delta/\mathcal{C}, \quad
        \sup_{z \in \mathbb{T}} |h_j(z) - \hat{h}_j(z)| \leq \delta/\mathcal{C}.
    \end{equation}
    Then, the degree required to \(\delta\)-approximate \(F'(x,z)\) increases slightly, but the overall scaling is unaffected.
    In other words,
    \begin{equation}
        \deg(\hat{F}') 
        = \bigl(O\bigl(\sqrt{n_l \ln(\mathcal{C}/\delta)}\bigr), O\bigl(\sqrt{m_l \ln(\mathcal{C}/\delta)}\bigr)\bigr) 
        = \bigl(O\bigl(\sqrt{n_l \ln(1/\delta)}\bigr), O\bigl(\sqrt{m_l \ln(1/\delta)}\bigr)\bigr),
    \end{equation}
    as \(\mathcal{C}\) is a constant.
    Therefore, the total error is upper bounded by \(\delta/\mathcal{C} \cdot \mathcal{C} \leq \delta\), which completes the proof. \\
\end{proof}

\textbf{Theorem~\ref{thm:product}}
\emph{
    Let \(G: [-1,1] \times \mathbb{R} \rightarrow \mathbb{R}\) be a product of elements
    \(f_i \in \mathcal{F}_{\mathrm{all}}\) and \(g_j \in \mathcal{G}_{\mathrm{all}}\), i.e.,
    \begin{equation}
        G(x,y) = \prod_{i=1}^{\mathcal{N}} f_i(x) \prod_{j=1}^{\mathcal{M}} g_j(y),
    \end{equation}
    where \(\mathcal{N} \geq 0\) and \(\mathcal{M} \geq 0\) stand for the constant number of functions.
    Let \(h_j(z)\) denote the Laurent polynomial representation of each \(g_j(y)\).
    The degree of \(G\) is the tuple \((n_p,m_p)\), given by the sums 
    \begin{equation}
        n_p \coloneqq \sum_{i=1}^{\mathcal{N}} \deg(f_i), \quad m_p \coloneqq \sum_{j=1}^{\mathcal{M}} \deg(h_j).
        \label{eq:appendix_prod_degree}
    \end{equation}
    Then, for any \(\delta \in (0, O(1)]\), 
    the function \(G\) can be \(\delta\)-approximated by a polynomial of degree
    \(\bigl(O\bigl(\sqrt{n_p \ln(1/\delta)}\bigr), O\bigl(\sqrt{m_p \ln(1/\delta)}\bigr)\bigr)\).
}
\begin{proof}
    Let \(G'(x,z)\) denote the Laurent polynomial representation of \(G(x,y)\):
    \begin{equation}
        G(x,y) = G'(x,z) \coloneqq \prod_{i=1}^{\mathcal{N}} f_i(x) \prod_{j=1}^{\mathcal{M}} h_j(z).
    \end{equation}
    Let \(\hat{G}'\) denote a \(\delta\)-approximation to \(G'\), defined by
    \begin{equation}
        \hat{G}'(x,z) \coloneqq \prod_{i=1}^{\mathcal{N}} \hat{f}_i(x) \prod_{j=1}^{\mathcal{M}} \hat{h}_j(z),
    \end{equation}
    where \(\hat{f}_i\) and \(\hat{h}_j\) are approximations to \(f_i\) and \(h_j\), respectively.
    We note that \(|f_i| \leq 1\) for all \(x \in [-1,1]\) and \(|h_j| \leq 1\) for all \(z \in \mathbb{T}\), 
    and assume that the approximation errors are upper bounded by \(\delta/(\mathcal{N} + \mathcal{M})\), i.e.,
    \begin{equation}
        \sup_{x \in [-1,1]} |f_i(x) - \hat{f}_i(x)| \leq \frac{\delta}{\mathcal{N} + \mathcal{M}}, \quad 
        \sup_{z \in \mathbb{T}} |h_j(z) - \hat{h}_j(z)| \leq \frac{\delta}{\mathcal{N} + \mathcal{M}}.
        \label{eq:prod_d_approx}
    \end{equation}
    Then \(\deg(\hat{G}') = \bigl(O\bigl(\sqrt{n_p \ln((\mathcal{N}+\mathcal{M})/\delta)}\bigr), O\bigl(\sqrt{m_p \ln((\mathcal{N}+\mathcal{M})/\delta)}\bigr)\bigr)\),
    where \(n_p\) and \(m_p\) are defined in Eq.~\eqref{eq:appendix_prod_degree}.
    Because \(\mathcal{N}\) and \(\mathcal{M}\) are constants, this factor does not affect the scaling, and we obtain
    \(\deg(\hat{G}') = \bigl(O\bigl(\sqrt{n_p \ln(1/\delta)}\bigr), O\bigl(\sqrt{m_p \ln(1/\delta)}\bigr)\bigr)\).
    To upper bound the approximation error \(|G' - \hat{G}'|\), we proceed inductively, starting with the first term as follows:
    \begin{equation}
    \begin{split}
        \sup_{x \in [-1,1], z \in \mathbb{T}} |G'(x,z) - \hat{G}'(x,z)| 
        & = \sup_{x \in [-1,1], z \in \mathbb{T}} 
        \left| \prod_{i=1}^{\mathcal{N}} f_i(x) \prod_{j=1}^{\mathcal{M}} h_j(z) 
        - \prod_{i=1}^{\mathcal{N}} \hat{f}_i(x) \prod_{j=1}^{\mathcal{M}} \hat{h}_j(z) \right| \\
        &
        \begin{aligned}
            \leq \sup_{x \in [-1,1], z \in \mathbb{T}} 
            & \left| \prod_{i=1}^{\mathcal{N}} f_i(x) \prod_{j=1}^{\mathcal{M}} h_j(z) 
            - \hat{f}_1(x) \prod_{i=2}^{\mathcal{N}} f_i(x) \prod_{j=1}^{\mathcal{M}} h_j(z) \right. \\
            & + \left. \hat{f}_1(x) \prod_{i=2}^{\mathcal{N}} f_i(x) \prod_{j=1}^{\mathcal{M}} h_j(z)
            - \prod_{i=1}^{\mathcal{N}} \hat{f}_i(x) \prod_{j=1}^{\mathcal{M}} \hat{h}_j(z) \right|
        \end{aligned} \\
        & 
        \begin{aligned}
            \leq & \sup_{x \in [-1,1], z \in \mathbb{T}} 
            \left| \prod_{i=1}^{\mathcal{N}} f_i(x) \prod_{j=1}^{\mathcal{M}} h_j(z) 
            - \hat{f}_1(x) \prod_{i=2}^{\mathcal{N}} f_i(x) \prod_{j=1}^{\mathcal{M}} h_j(z) \right| \\
            & + \sup_{x \in [-1,1], z \in \mathbb{T}} 
            \left| \hat{f}_1(x) \prod_{i=2}^{\mathcal{N}} f_i(x) \prod_{j=1}^{\mathcal{M}} h_j(z) 
            - \prod_{i=1}^{\mathcal{N}} \hat{f}_i(x) \prod_{j=1}^{\mathcal{M}} \hat{h}_j(z) \right| \\
        \end{aligned} \\
        & \leq \frac{\delta}{\mathcal{N}+\mathcal{M}} 
        + \sup_{x \in [-1,1], z \in \mathbb{T}} 
        \left| \hat{f}_1(x) \prod_{i=2}^{\mathcal{N}} f_i(x) \prod_{j=1}^{\mathcal{M}} h_j(z) 
        - \prod_{i=1}^{\mathcal{N}} \hat{f}_i(x) \prod_{j=1}^{\mathcal{M}} \hat{h}_j(z) \right|,
    \end{split}
    \end{equation}
    where the last inequality follows from Eq.~\eqref{eq:prod_d_approx}.
    Repeating this procedure \((\mathcal{N} + \mathcal{M} - 1)\) more times yields
    \begin{equation}
        \sup_{x \in [-1,1], z \in \mathbb{T}} \left|G'(x,z) - \hat{G}'(x,z)\right| 
        \leq \underbrace{\frac{\delta}{\mathcal{N}+\mathcal{M}} + \cdots + \frac{\delta}{\mathcal{N}+\mathcal{M}}}_{(\mathcal{N}+\mathcal{M}) \text{ terms}} \leq \delta.
    \end{equation}
    where the second inequality follows since \(\mathcal{N} + \mathcal{M} = O(1)\).
\end{proof}

\section{Review of QSP}
\label{appendix:review_qsp}
In this appendix, we review and provide the formal theorems of QSP used in Section~\ref{sec:synthesizing}.
The QSP theorem is stated as follows:
\begin{theorem}
    (Quantum Signal Processing in \(\mathrm{SU}(2)\) \cite[Theorem 3]{gilyen2019quantum}).
    \label{thm:qsp}
    Let \(d \in \mathbb{N}\). Then there exists a set of phase factors \(\Phi := (\phi_0, \cdots, \phi_d) \in [-\pi, \pi)^{d+1}\) such that
    \begin{equation}
        \label{eq:qsp}
        \begin{aligned}
                U_\Phi(x) = \rme^{\mathrm{i} \phi_0 \sigma_z} 
                \prod_{j=1}^{d} \left[ W(x) \rme^{\mathrm{i} \phi_j \sigma_z} \right]
                = \left( \begin{array}{cc}
                P(x) & \mathrm{i} Q(x) \sqrt{1 - x^2}\\
                \mathrm{i} Q^*(x) \sqrt{1 - x^2} & P^*(x)
                \end{array} \right),
        \end{aligned}
    \end{equation}
    where \(x \in [-1,1]\), and 
    \begin{equation}
        W(x) = \rme^{\mathrm{i} \arccos(x) \sigma_x} = \left(
            \begin{array}{cc}
                {x} & {\mathrm{i} \sqrt{1-x^{2}}} \\ 
                {\mathrm{i} \sqrt{1-x^{2}}} & {x}\end{array}
        \right),
    \end{equation}
    if and only if:
    \begin{enumerate}[label={(\arabic*)}]
        \item \(P, Q \in \mathbb{C}[x]\) with \(\deg(P) \leq d\) and \(\deg(Q) \leq d-1\).
        \item \(P\) has parity \((d \bmod 2)\) and \(Q\) has parity \((d-1 \bmod 2)\).
        \item \(|P(x)|^2 + (1-x^2) |Q(x)|^2 = 1\), \(\forall x \in [-1, 1]\).
    \end{enumerate}
\end{theorem}
Two types of operations are used in Eq.~\eqref{eq:qsp}: the signal operator \(W(x)\) and sequence of signal-processing operators \(\{\rme^{\mathrm{i}\phi_j \sigma_z}\}_{j=0}^{d}\) \cite{motlagh2024generalized}.
When the signal \(x\) is a matrix \(\ham\) of dimension \(2^N\), the block-encoding \(U_{\ham}\) and \(\{\rme^{\mathrm{i} \phi_j U_{\Pi}}\}_{j=0}^{d}\) correspond to the signal operator and signal-processing operators, 
respectively \cite{dong2021efficient}, where \(U_{\Pi}\) is defined as
\begin{equation}
    U_{\Pi} \coloneqq 2 \ket{0^M} \bra{0^M} \otimes I_N - I_{N+M}.
    \label{eq:cheby_spo}
\end{equation}
Here, \(M\) is the required number of ancilla qubits for block-encoding.
By noticing this correspondence, we can obtain the following QSP sequence that implements a block-encoded matrix polynomial \(P(\ham)\), provided \(P\) satisfies the three conditions in Theorem~\ref{thm:qsp} (for details, see \cite[Section II B]{dong2021efficient}):
\begin{equation}
    U_{\tilde{\Phi}} 
    = (-\mathrm{i})^d \left[ \prod_{j=0}^{d-1} \left( \rme^{\mathrm{i}\varphi_j U_{\Pi}} U_{\ham} \right) \right] \rme^{\mathrm{i} \varphi_d U_{\Pi}}.
    \label{eq:qsp_matrix}
\end{equation}
Note that the phase angles \(\varphi_j\) in Eq.~\eqref{eq:qsp_matrix} are different from \(\phi_j\) in Eq.~\eqref{eq:qsp}.
The relationship between them is given by:
\begin{equation}
    \varphi_j = \begin{cases}
        \phi_0 + \pi/4 & \quad (j = 0), \\
        \phi_j + \pi/2 & \quad (1 \leq j \leq d - 1), \\
        \phi_d + \pi/4 & \quad (j = d).
    \end{cases}
    \label{eq:varphi}
\end{equation}

In this work, we focus on implementing the polynomial-approximable functions that are real.
For implementing real polynomials, the real QSP theorem is employed:
\begin{theorem}
    (Real Quantum Signal Processing \cite[Corollary 10]{gilyen2019quantum}).
    \label{thm:real_qsp}
    Let \(\Re(P)\) denote the real part of a complex polynomial \(P \in \mathbb{C}[x]\) of degree \(d \geq 1\).
    If \(\Re(P)\) satisfies
    \begin{enumerate}[label={(\arabic*)}]
        \item \(\Re(P)\) has parity \((d \bmod 2)\), and
        \item \(|\Re(P(x))| \leq 1, \forall x \in [-1,1]\),
    \end{enumerate}
    then there exists a polynomial \(P \in \mathbb{C}[x]\) that satisfies all conditions in Theorem~\ref{thm:qsp}.
\end{theorem}
The process of implementing a real polynomial \(f(\ham)\) of degree \(d\) by Theorems~\ref{thm:qsp} and~\ref{thm:real_qsp} is as follows \cite{dong2021efficient}.
Let \(P^{*}(x)\) denote the complex conjugate of \(P(x)\), so that \(f = \Re(P(x)) = \frac{P(x) + P^{*}(x)}{2}\).
If \(\Re(P)\) satisfies the conditions in Theorem~\ref{thm:real_qsp}, \(P(x)\) satisfying the QSP conditions exists, and we can compute phase factors \(\phi_j\) for \(P(x)\).
They can be efficiently computed via numerical optimization with a computational cost of \(O(d^2)\) \cite[Algorithms 1 and 2]{dong2021efficient}.
From these phase factors \(\phi_j\), the corresponding phases \(\varphi_j\) for \(P(\ham)\) can be obtained using Eq.~\eqref{eq:varphi}.
The phase factors \(\varphi'_j\) for \(P^{*}(\ham)\) are given by \(\varphi'_j = -\varphi_j + \pi(1 - \delta_{jd})\) for \(j = 0, \ldots, d\).
Then, we achieve the implementation of \(f(\ham) = \Re(P(\ham))\) by a linear combination of unitaries (LCU) process of the two QSP sequences implementing \(P(\ham)\) and \(P^{*}(\ham)\).

By exploiting the relationship between \(\varphi_j\) and \(\varphi'_j\), the LCU process can be carried out without introducing an additional ancilla qubit, as shown in Fig.~\ref{fig:cqsp} 
(see \cite[Appendix B]{dong2021efficient} and \cite[Corollary 18]{gilyen2019quantum}).
Each \(\varphi'_j\) is obtained by negating \(\varphi_j\) and adding \(\pi\) when \(j < d\).
It corresponds to preparing the signal-processing register in the \(\ket{1}\) state and applying a Pauli-\(Z\) gate for each \(j < d\).
In conclusion, the circuit shown in Fig.~\ref{fig:cqsp} can implement a real polynomial \(f(\ham)\) of degree \(d\) and parity \((d \bmod 2)\), provided that \(|f| \leq 1\).

\begin{figure}[t]
    \begin{quantikz}[thin lines, row sep=0.3cm, column sep=0.3cm]
        \lstick{\(\ket{0}\)}
        & \gate{H} & \targ{} & \gate{\rme^{-\mathrm{i}\varphi_{d}\sigma_z}}   & \targ{} &                         
        & \gate{Z} & \targ{} & \gate{\rme^{-\mathrm{i}\varphi_{d-1}\sigma_z}} & \targ{} && \ \ldots\ 
        & \gate{Z} & \targ{} & \gate{\rme^{-\mathrm{i}\varphi_{0}\sigma_z}}   & \targ{} & \gate{H} & \meter{}\\
        \lstick{\(\ket{0^M}\)}
        && \octrl{-1} && \octrl{-1} & \gate[2]{U_{\ham}}  
        && \octrl{-1} && \octrl{-1} & \gate[2]{U_{\ham}} & \ \ldots\ 
        && \octrl{-1} && \octrl{-1} && \meter{}\\
        \lstick{\(\ket{\psi}\)}
        &&&&&&&&&&& \ \ldots\ &&&&&& \\
    \end{quantikz}
    \caption{
        The QSP circuit for implementing a real polynomial \(f(\ham)\) of degree \(d\) using a block-encoding \(U_{\ham}\) \cite[Fig. 16]{dong2021efficient}.
        Each signal-processing operator \(\rme^{\mathrm{i} \varphi_j U_{\Pi}}\) is realized using a gate sequence consisting of an \((M+1)\)-qubit Toffoli gate, a single-qubit \(Z\)-rotation, 
        and another \((M+1)\)-qubit Toffoli gate.
    }
    \label{fig:cqsp}
\end{figure}
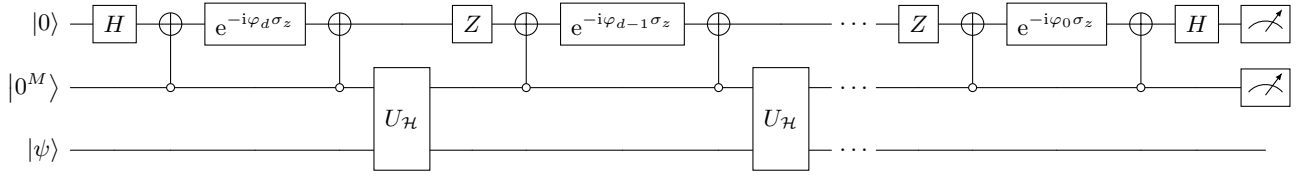

\section{Review of GQSP}
\label{appendix:review_gqsp}
In GQSP, controlled operations of unitary operators \(U \coloneqq \mathrm{e}^{\mathrm{i}\ham}\) and \(U^{\dagger} \coloneqq \mathrm{e}^{-\mathrm{i}\ham}\) are the signal operators so that a block-encoding of \(\ham\) is not used \cite{de2022fourier,motlagh2024generalized,haah2019product,berntson2025complementary}.
In this work, we we focus on the Laurent polynomials with real coefficients and adopt the Laurent formulation of QSP in Ref.~\cite{berntson2025complementary}.
Note that the Laurent polynomial approximations to sine functions may yield polynomials with purely imaginary coefficients (Appendix~\ref{appendix:sine_ver}).
As a global factor of \(\mathrm{i}\) can be factored out, we treat them as real-valued functions for circuit implementation.

Let \(L(z)\) be a Laurent polynomial of degree \(d \geq 1\) with real coefficients \(p_k \in \mathbb{R}\), defined on the unit circle \(\mathbb{T}\):
\begin{equation}
    L(z) = \sum_{k=-d}^{d} p_k z^k.
    \label{eq:laurent}
\end{equation}
The signal operator used in GQSP is given by:
\begin{equation}
    A = (\ket{0} \bra{0} \otimes U) + (\ket{1} \bra{1} \otimes U^{\dagger}) =
    \begin{bmatrix}
        U & 0 \\
        0 & U^{\dagger}
    \end{bmatrix}.
    \label{eq:lqsp_so}
\end{equation}
To implement a Laurent polynomial, specific SU(2) rotations are applied to an ancilla qubit \cite{berntson2025complementary}:
\begin{equation}
    R(\theta) = \begin{bmatrix}
        \cos(\theta) & \mathrm{i} \sin(\theta) \\
        \mathrm{i} \sin(\theta) & \cos(\theta) \\
    \end{bmatrix} \otimes \mathbb{I}.
    \label{eq:lqsp_spo}
\end{equation}
Following \cite[Theorem 6]{berntson2025complementary}, \cite[Theorem 3]{motlagh2024generalized}, and \cite{haah2019product}, we state the GQSP theorem as follows:
\begin{theorem}
    (General Quantum Signal Processing for Implementing Laurent Polynomials \cite{haah2019product}).
    \label{thm:lqsp}
    Let \(L\) be a Laurent polynomial of degree \(d \in \mathbb{N}\) with real coefficients, as given in Eq.~\eqref{eq:laurent}.
    Assume that \(L\) has parity \(d \bmod 2\), and satisfies \(|L(z)| \leq 1\) for all \(z \in \mathbb{T}\).
    Then, for all \(z \in \mathbb{T}\), there exists a complementary Laurent polynomial \(K \in \mathbb{R}[z, z^{-1}]\) 
    and a sequence of angles \((\theta_j)_{j=0}^{d} \in (-\pi,\pi]^{d+1}\) such that
    \begin{equation}
        |L(z)|^2 + |K(z)|^2 = 1,
    \end{equation}
    and
    \begin{equation}
        \begin{bmatrix}
            L(U)        & \mathrm{i} K(U)  \\
            -\mathrm{i} K(U^{\dagger}) & L(U^{\dagger}) \\
        \end{bmatrix} = 
        R(\theta_0) \left( \prod_{j=1}^{d} A \cdot R(\theta_{j}) \right).
    \end{equation}
\end{theorem}
To determine the parameters \(\theta_j\), one first computes the polynomial \(K\) from \(L\) using FFT-based convolution algorithms.
The angles \(\theta_j\) can then be recursively obtained from the ratios of the coefficients of \(L\) and \(K\) (see \cite[Algorithm 1]{motlagh2024generalized}). 
The circuit implementation of GQSP is shown in Fig.~\ref{fig:lqsp}.

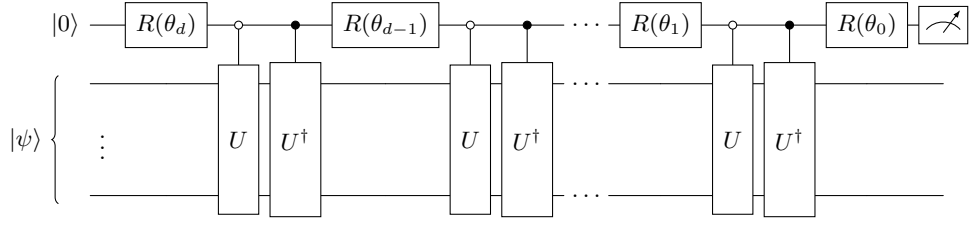
\begin{figure}[t]
    \begin{quantikz}[thin lines, row sep=0.2cm, column sep=0.15cm]
        \lstick{\(\ket{0}\)}
        & \gate{R(\theta_d)} & \octrl{1} & \ctrl{1}
        & \gate{R(\theta_{d-1})} & \octrl{1} & \ctrl{1} & \ \ldots\
        & \gate{R(\theta_{1})} & \octrl{1} & \ctrl{1} & \gate{R(\theta_{0})} & \meter{} \\
        \lstick[3]{\(\ket{\psi}\)}
        && \gate[3]{U} & \gate[3]{U^{\dagger}}
        && \gate[3]{U} & \gate[3]{U^{\dagger}} & \ \ldots\ 
        && \gate[3]{U} & \gate[3]{U^{\dagger}} && \\
        \quad \ \vdots\ \\
        &&&&&&& \ \ldots\ &&&&&
    \end{quantikz}
    \caption{
        GQSP circuit for implementing a Laurent polynomial \(L(U)\) with real coefficients.
        The rotation gate \(R(\theta)\) is defined in Eq.~\eqref{eq:lqsp_spo}.
    }
    \label{fig:lqsp}
\end{figure}

\section{Depth Analysis of GQSP Circuits Implementing Laurent Polynomials}
\label{appendix:analysis_gqsp_circ}
For the Laurent polynomial representations, their GQSP circuits involve both \(U\) and \(U^{\dagger}\).
These unitaries are implemented using the \(2v\)th-order symmetric Suzuki-Trotter decomposition \(S_{2v}(t)\) defined recursively for \(v \geq 2\):
\begin{equation}
    S_{2v}(t) \coloneqq S_{2v-2}(u_{v}t)^2\, S_{2v-2}((1-4u_v)t) S_{2v-2}(u_{v}t)^2,
    \label{eq:high_st}
\end{equation}
where \(u_v \coloneqq 1/(4 - 4^{1/(2v-1)})\), and the second-order symmetric Suzuki-Trotter formula \(S_2(t)\) is given by
\begin{equation}
    S_{2}(t) \coloneqq 
    \rme^{\eta\ham_0} \rme^{\eta\ham_1} \cdots \rme^{\eta\ham_{J-1}} 
    \rme^{\eta\ham_{J-1}} \cdots \rme^{\eta\ham_1} \rme^{\eta\ham_0},
\end{equation}
with \(\eta = -\mathrm{i}t/2\) \cite{suzuki1991general}.
Theorem 6 in Ref.~\cite{childs2021theory} 
states that the additive and multiplicative Trotter errors \(\mathcal{E}_{\text{ST}}\) of \(S_{2v}(t)\) for a normalized Hamiltonian are both upper bounded by \(\mathcal{E}_{\text{ST}} = O(t^{2v+1})\).
To reduce the Trotter errors, \(S_{2v}(t)\) is fragmented with \(r\) repetitions of \(S_{2v}(t/r)\), where \(r\) is the number of Trotter steps. 

A Laurent polynomial of degree \(d_2\) introduces \(O(d_2^2)\) additive and multiplicative Trotter errors.
Therefore, for \(t = 1\), the number of Trotter steps \(r\) should be chosen as
\begin{equation}
    r = O\left(\frac{d_2^{1/v}}{\delta^{1/(2v)}}\right),
\end{equation}
to ensure that the total Trotter error in the \(d_2\)-degree Laurent polynomial is upper bounded by \(O(\delta)\) \cite[Corollary 12]{childs2021theory}.

\(S_{2v}(1/r)\) consists of a sequence of \(2 \cdot 5^{v-1}\) applications of \(S_2\), as defined in Eq.~\eqref{eq:high_st}, each of which is decomposed into \(O(Jk)\) CNOTs \cite{nielsen2010quantum}.
Therefore, the circuit depth of a GQSP sequence for a degree-\(d_2\) Laurent polynomial is
\begin{equation}
    O(d_2 \cdot r \cdot 2 \cdot 5^{v-1} \cdot Jk) = O\left(d_2^{1+1/v} \cdot D_{\text{ST}} \right),
\end{equation}
where we define \(D_{\text{ST}} \coloneqq 5^{v-1} Jk/\delta^{1/(2v)}\).

\end{document}